\documentclass[journal]{IEEEtran}
\usepackage{amsmath,graphicx}

\usepackage{amsmath,amssymb,epsfig,psfrag,cite,subfigure}
\include{macros}
\usepackage{graphicx}
\usepackage{epstopdf}

\usepackage{bm}

\usepackage{geometry}
\usepackage{tcolorbox}

\usepackage{soul}

\usepackage{amsthm}

\usepackage{nicefrac}

\usepackage{lipsum,multicol}

\usepackage{algorithm}
\usepackage{algpseudocode}

\usepackage{mwe}    
\usepackage{epsfig,psfrag}
\usepackage{subfigure}
\usepackage{color}
\usepackage{url}
\usepackage{mathtools,xparse}

\usepackage{gensymb}

\usepackage[multiple]{footmisc}

\theoremstyle{remark}

\newtheoremstyle{mytheoremstyle} 
    {\topsep}                    
    {\topsep}                    
    {\upshape}                   
    {.5em}                           
    {\itshape}                   
    {.}                          
    {.5em}                       
    {}  
    
\theoremstyle{mytheoremstyle}

\newtheoremstyle{iremark}
  {\topsep}   
  {\topsep}   
  {\upshape}  
  {0.2in}       
  {\itshape}  
  {.}         
  {5pt plus 1pt minus 1pt} 
  {\thmname{#1}\thmnumber{ \itshape#2}\thmnote{ (#3)}} 

\theoremstyle{iremark}

\newtheorem{proposition}{Proposition}
\newtheorem*{proof}{Proof}

\DeclareFontFamily{U}{mathx}{\hyphenchar\font45}
\DeclareFontShape{U}{mathx}{m}{n}{
	<5> <6> <7> <8> <9> <10>
	<10.95> <12> <14.4> <17.28> <20.74> <24.88>
	mathx10
}{}
\DeclareSymbolFont{mathx}{U}{mathx}{m}{n}
\DeclareFontSubstitution{U}{mathx}{m}{n}

\DeclarePairedDelimiter\abs{\lvert}{\rvert}%

\DeclarePairedDelimiter{\nint}\lfloor\rceil

\makeatletter \renewcommand\d[1]{\ensuremath{%
		\;\mathrm{d}#1\@ifnextchar\d{\!}{}}}
\makeatother

\makeatletter
\newcommand*\rel@kern[1]{\kern#1\dimexpr\macc@kerna}
\newcommand*\widebar[1]{%
  \begingroup
  \def\mathaccent##1##2{%
    \rel@kern{0.8}%
    \overline{\rel@kern{-0.8}\macc@nucleus\rel@kern{0.2}}%
    \rel@kern{-0.2}%
  }%
  \macc@depth\@ne
  \let\math@bgroup\@empty \let\math@egroup\macc@set@skewchar
  \mathsurround\z@ \frozen@everymath{\mathgroup\macc@group\relax}%
  \macc@set@skewchar\relax
  \let\mathaccentV\macc@nested@a
  \macc@nested@a\relax111{#1}%
  \endgroup
}
\makeatother

\newcommand{\rev}[1]{\textcolor{black}{#1}} 
\newcommand{\revv}[1]{\textcolor{black}{#1}} 

\newcommand{\yy}{\mathbf{y}}

\newcommand{\uu}{\mathbf{u}}

\newcommand{\nn}{\mathbf{n}}
\newcommand{\pp}{\mathbf{p}}
\newcommand{\qq}{\mathbf{q}}
\newcommand{\rr}{\mathbf{r}}

\newcommand{\yklmbar}{ \widebar{\yy}_{k,\ell,m} }

\newcommand{\Imatrix}{{ \boldsymbol{\mathrm{I}} }}

\newcommand{\fb}{\mathbf{f}}

\newcommand{\ekk}[1]{\mathbf{e}_{#1}}
\newcommand{\ekkt}[1]{\mathbf{e}^T_{#1}}

\newcommand{\WW}{\mathbf{W}}

\newcommand{\UU}{\mathbf{U}_{\rmtx}}

\newcommand{\Lambdab}{\mathbf{\Lambda}}

\newcommand{\HH}{\mathbf{H}}
\newcommand{\FF}{\mathbf{F}}

\newcommand{\XXb}{\mathbf{X}}

\newcommand{\SSbig}{\mathbf{S}}

\newcommand{\Pmax}{P_{\rm{max}}}

\newcommand{\XXbstar}{\mathbf{X}^\star}

\newcommand{\projrange}[1]{\boldsymbol{\Pi}_{#1}}
\newcommand{\projnull}[1]{\boldsymbol{\Pi}^{\perp}_{#1}}

\usepackage{accents}
\newcommand*{\dt}[1]{%
	\accentset{\mbox{\large .}}{#1}}

\newcommand{\FFh}{\mathbf{F}^H}

\newcommand{\FFsum}{\FF^{\rm{sum}}}
\newcommand{\FFdiff}{\FF^{\rm{diff}}}
\newcommand{\FFopt}{\FF^{\rm{opt}}}
\newcommand{\FFdifftilde}{\widetilde{\FF}^{\rm{diff}}}
\newcommand{\FFanalog}{\FF^{\rm{analog}}}
\newcommand{\FFdigital}{\FF^{\rm{digital}}}

\newcommand{\JJ}{\mathbf{J}}
\newcommand{\TT}{\mathbf{T}}
\newcommand{\JJpos}{\widetilde{\JJ}}

\newcommand{\Frx}{ F_{\rmrx} }

\newcommand{\JJprior}{\mathbf{J}^{\rm{prior}}}

\newcommand{\fpeb}{{\rm{PEB}}}

\newcommand{\atantwo}{{\rm{atan2}}}
\newcommand{\sigmaclk}{ \sigma_{\rm{clk}}^2 }
\newcommand{\sigmaclkk}{ \sigma_{\rm{clk}} }
\newcommand{\JJclk}{ J_{\rm{clk}} }

\newcommand{\tracesmall}[1]{ {{{\rm{tr}}\left( #1 \right)}}  }

\newcommand{\AAb}{\mathbf{A}}

\newcommand{\AAtxthetai}{\AAb_{\rmtx} }
\newcommand{\AAtxthetaiconj}{\AAb^{*}_{\rmtx} }
\newcommand{\AAtxthetaih}{\AAb^T_{\rmtx} }

\newcommand{\AAtxthetadtalli}{\dt{\AAb}_{{\rmtx}} }

\newcommand{\AAtxthetadtallih}{\dt{\AAb}_{{\rmtx}}^T }

\newcommand{\atxdt}{\dt{\mathbf{a}}_{{\rmtx}} }
\newcommand{\atxdttilde}{\widetilde{\dt{\mathbf{a}}}_{{\rmtx}} }

\newcommand{\norm}[1]{\left\lVert#1\right\rVert}

\newcommand{\boldone}{{ {\boldsymbol{1}} }}

\newcommand{\Ptot}{{ P_{\rm{tot}} }}
\newcommand{\etab}{{ \bm{\eta} }}
\newcommand{\rhob}{{ \bm{\rho} }}
\newcommand{\etabtilde}{{ \bm{\widetilde{\eta}} }}

\newcommand{\skl}{s_{k,\ell}}
\newcommand{\nklm}{\nn_{k,\ell,m}}

\newcommand{\rmtx}{{\rm{Tx}}}
\newcommand{\rmrx}{{\rm{Rx}}}

\newcommand{\rmrf}{{\rm{RF}}}

\newcommand{\deltaf}{ \Delta f }
\newcommand{\deltat}{ \Delta t }

\newcommand{\Ntx}{N_\rmtx}
\newcommand{\Nrx}{N_\rmrx}
\newcommand{\Nrxrf}{N_{\rmrf}}

\newcommand{\arx}{\mathbf{a}_\rmrx}
\newcommand{\atx}{\mathbf{a}_\rmtx}

\newcommand{\diag}[1]{ {\rm{diag}}\left(#1\right)  }

\newcommand{\realp}[1]{ \Re \left\{#1\right\}  }
\newcommand{\imp}[1]{ \Im \left\{#1\right\}  }

\newcommand{\thetab}{ \bm{\theta} }
\newcommand{\phib}{ \bm{\phi} }

\newcommand{\taub}{ \bm{\tau} }

\newcommand{\alphabr}{ \bm{\alpha}_{\rm{R}} }
\newcommand{\alphabi}{ \bm{\alpha}_{\rm{I}} }

\newcommand{\fc}{ f_c }

\newcommand{\complexset}[2]{ \mathbb{C}^{#1 \times #2}  }
\newcommand{\realset}[2]{ \mathbb{R}^{#1 \times #2}  }

\newcommand{\thnew}[1]{ {{#1^{\rm{th}}}} }

\newcommand{\mtCN}{{\mathcal{CN}}}
\newcommand{\boldzero}{{ {\boldsymbol{0}} }}

\graphicspath{{./Figures/}}

\begin{document}
\bstctlcite{IEEEexample:BSTcontrol}

\title{Optimal Spatial Signal Design for mmWave Positioning under Imperfect Synchronization}

\author{Musa Furkan Keskin, \textit{Member, IEEE}, Fan Jiang, \textit{Member, IEEE}, Florent Munier, \\Gonzalo Seco-Granados, \textit{Senior Member, IEEE}, Henk Wymeersch, \textit{Senior Member, IEEE}\thanks{Copyright (c) 2015 IEEE. Personal use of this material is permitted. However, permission to use this material for any other purposes must be obtained from the IEEE by sending a request to pubs-permissions@ieee.org.}
\thanks{M. F. Keskin, F. Jiang and H. Wymeersch are with the Department of Electrical Engineering, Chalmers University of Technology, Sweden. G. Seco-Granados is with the Department of Telecommunications and Systems Engineering, Universitat Autonoma de Barcelona, Spain. F. Munier is with Ericsson Research, Sweden. This work is supported, in part, by Vinnova 5GPOS project under grant 2019-03085, MSCA-IF grant
888913 (OTFS-RADCOM), the European Commission through the H2020 project Hexa-X (Grant Agreement no. 101015956) and ICREA Academia Program.}\vspace{-0.2in}}

\maketitle


\begin{abstract}
    We consider the problem of spatial signal design for multipath-assisted mmWave positioning under limited prior knowledge on the user's location and clock bias.
    \rev{We propose an optimal robust design and, based on the low-dimensional precoder structure under perfect prior knowledge, a codebook-based heuristic design with optimized beam power allocation.}
    Through numerical results, we characterize 
    different position-error-bound (PEB) regimes with respect to clock bias uncertainty and show that the proposed low-complexity codebook-based designs outperform the conventional directional beam codebook and achieve near-optimal PEB performance for both analog and digital architectures. 
    
    
\end{abstract}
\vspace{-5mm}
\section{Introduction}\label{sec_intro}
In 5G mmWave systems, in addition to time-of-arrival (TOA) measurements, the ability to estimate angles-of-arrival (AOAs) and angles-of-departure (AODs) has been introduced \cite{TR38.855}, which provides additional geometric information, increases multipath resolvability and enables the exploitation of reflected propagation paths to improve positioning\rev{, in contrast to sub-6 GHz systems}. Based on the combination of TOA, AOA and AOD measurements, 
joint positioning and synchronization with a single BS can be realized under multipath conditions \cite{Mendrzik_JSTSP_2019}.   
In contrast to positioning signals in time and frequency \cite[Section 7.4.1.7]{TS38.211}, spatial positioning signals are inherently directional and thus mainly meaningful under \textit{a-priori} information about the location of the user equipment (UE), which can be obtained in a \textit{tracking} scenario \cite{mmwave_training_2016,Mendrzik_JSTSP_2019}. Conventional spatial signal design at the BS involves a set of \textit{directional} beams (e.g., from a DFT codebook), where the reported received power at the UE reveals a coarse AOD information \cite{dwivedi2021positioning}.

With the advent of digital and hybrid arrays, more refined spatial designs that boost localization performance become possible 
\cite{precoderNil2018,successiveLocBF_2019,tasos_precoding2020}. In \cite{precoderNil2018}, the authors adopt the Cram\'{e}r-Rao bound\footnote{The CRB is a meaningful criterion in the presence of a priori location information, as then fine beams can be used with high SNR. } (CRB) as a design criterion and optimize the AOA and AOD accuracy under pure line-of-sight\footnote{Here, LOS refers to the presence of only a geometric direct path between BS and UE, while non-line-of-sight (NLOS) refers to the case with additional propagation paths, where the LOS path may or may not be present.} (LOS) propagation by exploiting prior knowledge on UE location. Similarly, \cite{successiveLocBF_2019} proposes an iterative signal design approach to minimize the CRB on location estimation in the presence of both LOS and NLOS paths, which requires multiple rounds of localization and signal optimization, thus causing significant overhead. The studies in \cite{precoderNil2018,successiveLocBF_2019} share two major limitations
, namely, \textit{(i)} the high-complexity of the \textit{unconstrained} design (i.e., without a specific codebook), and \textit{(ii)} the impractical assumption of \textit{perfect synchronization} between BS and UE \cite{Mendrzik_JSTSP_2019}. To address such limitations, \cite{tasos_precoding2020} considers a mmWave localization setup with clock offset between BS and UE, and designs \textit{codebook-based} spatial signals containing only conventional \textit{directional} beams. However, the prior approaches to signal design \cite{precoderNil2018,successiveLocBF_2019,tasos_precoding2020} leave undiscovered whether \textit{(i)} novel codebooks with alternative beams in addition to \textit{directional} ones can be devised to improve accuracy over standard codebooks, \textit{(ii)} the localization performance of \textit{codebook-based} designs can approach that of \textit{unconstrained} ones, and \textit{(iii)} \textit{analog} beams can achieve similar performance to \textit{digital} ones.

With the goal of tackling the above issues, we consider the design of both \textit{unconstrained} and \textit{codebook-based} spatial signals for 5G and beyond 5G multipath-assisted mmWave positioning under limited prior knowledge on UE location \cite{mmwave_training_2016,precoderNil2018,Mendrzik_JSTSP_2019} and \textit{imperfect synchronization} in \textit{analog} and \textit{digital} arrays. Our main contributions are: \textit{(i)} We derive the structure of the optimal spatial signals under perfect prior knowledge on the UE location \cite{li2007range}, establishing an insightful connection to monopulse radar \cite{monopulse_review}; \textit{(ii)} Inspired by the low-dimensional structure of these optimal signals, we propose a novel low-complexity codebook-based design under imperfect prior knowledge, involving both \textit{directional} and \textit{derivative} beams, and an optimal unconstrained design
; \textit{(iii)} We characterize different position-error-bound (PEB) regimes with respect to clock bias uncertainty, providing practical guidelines on which spatial signal to use under different levels of synchronization; \textit{(iv)} We show that the proposed low-complexity codebook can be  implemented in an analog 
architecture,  
outperforms the codebook with only directional beams \cite{tasos_precoding2020}, and attains the upper bound specified by the unconstrained design.\footnote{In other words, incorporating \textit{derivative} beams into the proposed codebook enables closing the performance gap between the unconstrained design and the codebook-based design.}

\vspace{-3mm}
\section{System Model and Problem Description}\label{sec_sysmodel}
\subsubsection*{Signal Model}
Consider a MIMO-OFDM mmWave downlink communications scenario with a BS and a \rev{single} UE \rev{(as in \cite{precoderNil2018,successiveLocBF_2019,
tasos_precoding2020})}, equipped with $\Ntx$ and $\Nrx$ antennas, respectively. \rev{Having an architecture with the capability of transmitting scaled and phase-shifted versions of signal across the different antennas\footnote{\label{fn_arch}\rev{This architecture is called \textit{digital} throughout the paper for the sake of simplicity (though it is sufficient to employ analog active phased arrays with controllable amplitude per antenna, since only single-stream pilots are considered). In addition, \textit{analog} architectures refer to standard analog passive arrays with no amplitude control per antenna.}},} the BS transmits $M$ identical pilot frames sequentially, where each frame $\SSbig \in \complexset{K}{L}$ consists of $L$ OFDM symbols with $K$ subcarriers, each representing a time-frequency comb pattern  
\cite{TR38.855,TS38.211}. Assuming that the $\thnew{m}$ pilot frame is precoded by a  vector $\fb_m \in \complexset{\Ntx}{1}$ for $m = 0, 1, \ldots, M-1$, the received signal vector at the UE over the $\thnew{k}$ subcarrier of the $\thnew{\ell}$ symbol in the $\thnew{m}$ frame can be expressed as \cite{precoderNil2018}
\begin{equation}\label{eq_rec}
    \yy_{k,\ell,m} = \WW^H \HH_k \fb_m \skl + \nklm  ~,
\end{equation}
for $k = 0, \ldots, K-1$, $\ell = 0, \ldots, L-1$, where
$\WW \in \complexset{\Nrx}{\Nrxrf}$ is the analog combining matrix at the UE, with $\Nrxrf$ denoting the number of RF chains, $\skl = \left[ \SSbig \right]_{k, \ell}$ is the pilot symbol on the $\thnew{k}$ subcarrier of the $\thnew{\ell}$ symbol, 
$\HH_k \in \complexset{\Nrx}{\Ntx}$ is the channel matrix at the $\thnew{k}$ subcarrier, given by 
\begin{equation} \label{eq_hk}
       \HH_k = \sum_{\revv{g}=0}^{\revv{G} -1} \alpha_{\revv{g}} \, e^{-j 2 \pi k \deltaf \tau_{\revv{g}} } \arx(\phi_{\revv{g}}) \atx^T(\theta_{\revv{g}})~,  
\end{equation}
and $\nklm \sim \mtCN(\boldzero, \sigma^2 \Imatrix )$ is additive white Gaussian noise. Here, $\deltaf$ is the subcarrier spacing, \revv{$G$} denotes the number of paths, $\alpha_{\revv{g}}$, $\tau_{\revv{g}}$, $\phi_{\revv{g}}$ and $\theta_{\revv{g}}$ are the complex channel gain, delay, AOA and AOD of the $\thnew{\revv{g}}$ path, respectively, and $\atx(\theta)  \in \complexset{\Ntx}{1}$ and $\arx(\phi) \in \complexset{\Nrx}{1}$ denote the array steering vectors at the BS and UE, respectively.

\subsubsection*{System Geometry}
We consider a 2D positioning setup \rev{(see Fig.\,\ref{fig_sys_geometry})} where the BS and UE are located at $\qq = \left[ q_x, q_y \right]^T \in \realset{2}{1}$ and $\pp = \left[ p_x, p_y \right]^T \in \realset{2}{1}$, respectively, and the orientation of the UE is denoted by $\psi \in [0, 2 \pi)$. We assume that $\qq$ is known, while $\pp$ and $\psi$ are unknown to be estimated from \eqref{eq_rec}, as in \cite{shahmansoori2017position,kakkavas2019performance}. For the channel model, ${\revv{g}}=0$ represents the LOS path, while ${\revv{g}} \geq 1$ correspond to the NLOS paths. Each NLOS path is associated with \revv{an incidence point} with unknown location $\rr_{\revv{g}} = \left[ r_{{\revv{g}},x}, r_{{\revv{g}},y} \right]^T \in \realset{2}{1}$. According to the system geometry, the angles can be expressed as 
\begin{align}
    \theta_{{\revv{g}}>0} & =\atantwo( r_{{\revv{g}},y} - q_y, r_{{\revv{g}},x} - q_x) ~,\\
    \theta_0& =\atantwo( p_y - q_y, p_x - q_x ) ~,\\
    \phi_{{\revv{g}}>0}& =\atantwo( p_y - r_{{\revv{g}},y}, p_x - r_{{\revv{g}},x} ) - \psi ~,\\
    \phi_0& =\theta_0 - \psi,
\end{align}
with $\atantwo(y,x)$ denoting the four-quadrant inverse tangent. The path delays are given by
\begin{align}
    \tau_{{\revv{g}}>0}& =(\norm{ \qq - \rr_{\revv{g}} }_2 + \norm{ \rr_{\revv{g}} - \pp }_2) /c + \deltat ~,\\
    \tau_{0}& =\norm{ \pp - \qq }_2/c + \deltat,
\end{align}
where $c$ is the speed of propagation and $\deltat$ is the clock bias between the BS and UE, modeled as $\deltat\sim p(\deltat)$.

\begin{figure}
\centering
    \vspace{-0.8in}
	\includegraphics[width=1.05\linewidth]{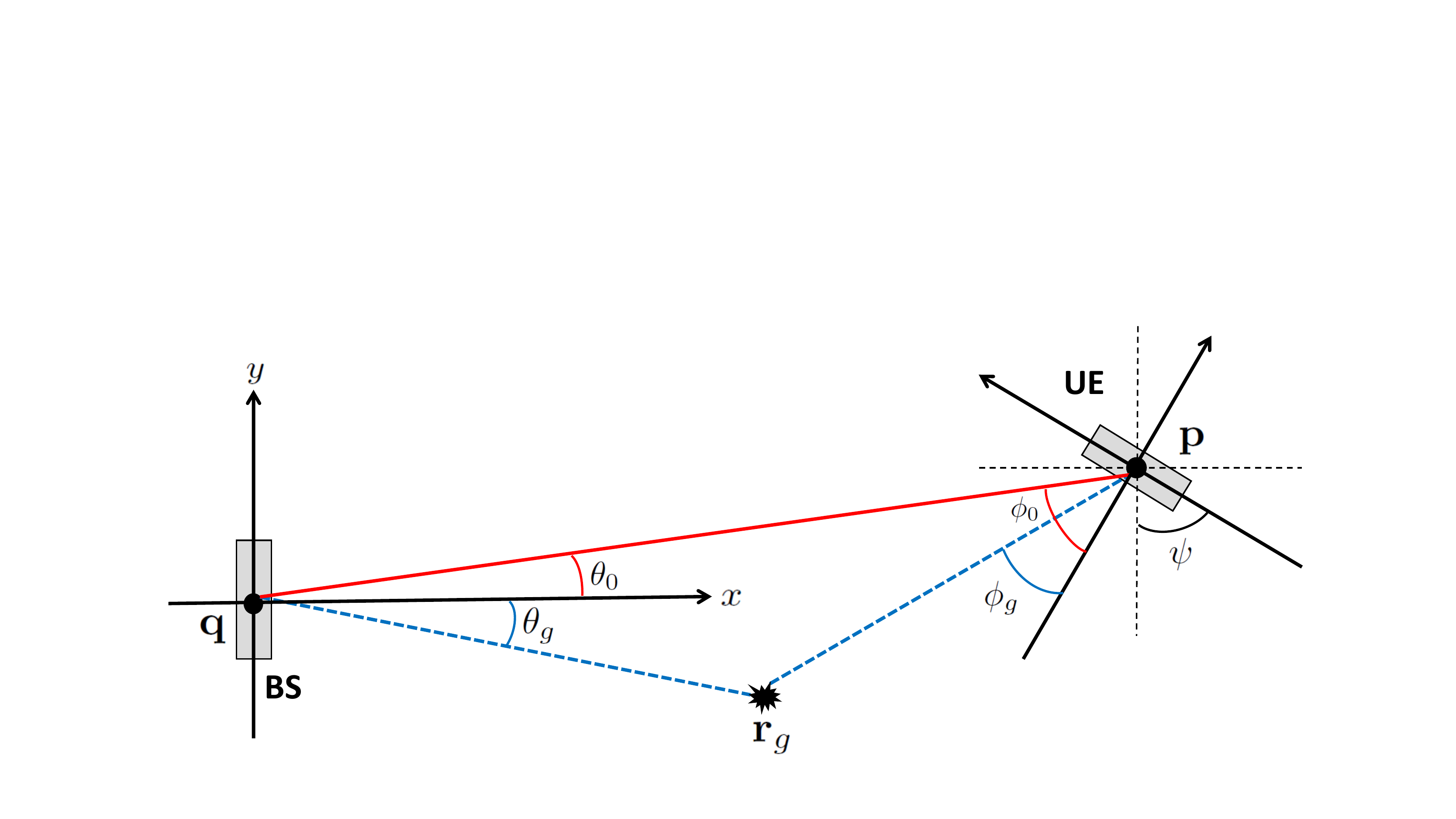}
	\vspace{-0.4in}
	\caption{\rev{2D geometric model of mmWave MIMO positioning setup with known BS position $\qq$, unknown UE location $\pp$ and orientation $\psi$, and unknown \revv{incidence point} location $\rr_{\revv{g}}$}.}
	\label{fig_sys_geometry}
	\vspace{-0.25in}
\end{figure}

\subsubsection*{Problem Description for Spatial Signal Design}
Our problem of interest for mmWave spatial signal design is to find the optimal (analog, hybrid, or digital) precoder $\FF  \triangleq \left[ \fb_0 \, \ldots \, \fb_{M-1} \right] \in \complexset{\Ntx}{M}$ that maximizes the accuracy of estimation of UE location $\pp$ in a setup where also the orientation $\psi$, the clock bias $\deltat$ and the \revv{incidence point} locations $\{ \rr_{\revv{g}} \}_{\revv{g}=1}^{\revv{G}-1}$ have to be estimated. The input data for the estimator are the observations $\{ \yy_{k,\ell,m} \}_{\forall k,\ell,m}$ in \eqref{eq_rec}, collected over $M$ beams each employing $L$ symbols with $K$ subcarriers, together with the prior knowledge\footnote{\label{fn_prior}The prior information can be obtained via joint tracking of UE and \revv{incidence point} positions and clock bias  \cite{Mendrzik_JSTSP_2019,precoderNil2018}. \rev{The prior knowledge on $\pp$ and $\{ \rr_{\revv{g}} \}_{{\revv{g}}=1}^{{\revv{G}}-1}$ will be described in Sec.~\ref{sec_robust}, while the effect of $\deltat$ on PEB will be introduced in Sec.~\ref{sec_crb_metric}.}} on $\pp$, $\{ \rr_{\revv{g}} \}_{{\revv{g}}=1}^{{\revv{G}}-1}$ and $\deltat$. \rev{Since the BS has no control over the UE, $\WW$ is assumed to be fixed and not subject to optimization.}

\vspace{-2mm}
\section{Spatial Signal Design for mmWave Positioning}\label{sec_perf_met}
In this section, we derive the CRB based performance metric for position estimation, and propose unconstrained and codebook-based strategies to design spatial signals.

\vspace{-2mm}
\subsection{CRB-Based Performance Metric}\label{sec_crb_metric}
The channel and location domain unknown parameter vectors are given by $\etab = \left[ \thetab^T, \phib^T,  \alphabr^T,  \alphabi^T,  \taub^T  \right]^T \in \realset{5\revv{G}}{1}$ and $\etabtilde = [   \pp^T, \psi, \rr^T, \alphabr^T, \alphabi^T, \deltat  ]^T \in \realset{(4\revv{G}+2)}{1}$, respectively,
where $\thetab = \left[ \theta_0, \ldots, \theta_{\revv{G}-1} \right]^T$, $\phib = \left[ \phi_0, \ldots, \phi_{\revv{G}-1} \right]^T$, $\alphabr = \left[ \realp{\alpha_0}, \ldots, \realp{\alpha_{\revv{G}-1}} \right]^T$, $\alphabi = \left[ \imp{\alpha_0}, \ldots, \imp{\alpha_{\revv{G}-1}} \right]^T$, $\taub = \left[\tau_0, \ldots, \tau_{\revv{G}-1} \right]^T$ and $\rr = [ \rr_1^T, \ldots, \rr_{\revv{G}-1}^T ]^T$. The Fisher information matrix (FIM) $\JJ \in \realset{5\revv{G}}{5\revv{G}}$ of $\etab$ can be computed from the signal model in \eqref{eq_rec} \rev{using the Slepian-Bangs formula \cite[Eq.~(15.52)]{kay2013fundamentals}} as
\begin{equation}\label{eq_jobs_ij}
    \left[ \JJ \right]_{i,j} = \frac{2}{\sigma^2}  \sum_{k=0}^{K-1} \sum_{\ell=0}^{L-1} \sum_{m=0}^{M-1} \realp{ \frac{\partial \yklmbar^{H}}{\partial \eta_i} \frac{\partial \yklmbar}{\partial \eta_j} }  
\end{equation}
for $i, j = 0, \ldots, 5\revv{G}-1$, where $\yklmbar = \WW^H   \HH_k \fb_m \skl$.

Since our focus is on the positioning performance, we derive the FIM $\JJpos \in \realset{(4\revv{G}+2)}{(4\revv{G}+2)}$ of $\etabtilde$ as 
\begin{align}\label{eq_jjpos}
    \JJpos = \TT^T \JJ \TT + \JJprior ~,
\end{align}
where $\TT^T \JJ \TT$ and $\JJprior$ represent the FIMs derived from the observations and the prior knowledge, respectively, and the transformation matrix $\TT \in \realset{5\revv{G}}{(4\revv{G}+2)}$ can be expressed as the Jacobian $[ \TT ]_{i,j} = \partial [ \etab ]_i / \partial [ \etabtilde ]_j$.
Since only the clock bias is assumed to be random in $\etabtilde$, we have  $\left[ \JJprior \right]_{(4\revv{G}+2,4\revv{G}+2)}=\JJclk \ge 0$,  derived from $p(\deltat)$, 
and all other entries of $\JJprior$ are zero \cite{kakkavas2019performance}. To quantify the position estimation accuracy, we adopt the PEB as our performance metric, which can be computed via \cite{shahmansoori2017position,kakkavas2019performance}
\begin{align}\label{eq_peb}
    \fpeb(\FF; \etabtilde) = \left[ \tracesmall{ [ \JJpos^{-1} ]_{0:1,0:1} }  \right]^{1/2} ~.
\end{align} 
\rev{As seen from Appendix~\ref{app_fim_der}, the FIM $\JJpos$ and the PEB in \eqref{eq_peb} are functions of the precoder $\FF$ employed by the BS.}

\vspace{-2mm}
\subsection{Optimal Signal Design with  Perfect Knowledge}\label{sec_opt_sig_perfect}
The PEB 
depends on the deterministic unknown parameters in $\etabtilde$ (i.e., $\pp$, $\psi$, $\rr$, $\alphabr$ and $\alphabi$). In this part, we assume perfect knowledge of these parameters and formulate the signal design problem accordingly. In Sec.~\ref{sec_robust}, we will focus on robust signal design in the presence of uncertainties
in $\etabtilde$. Under perfect knowledge of $\etabtilde$, the signal design problem with a total power constraint \revv{$\Ptot$} can be formulated as
\begin{align} \label{eq_problem_peb}
	\mathop{\mathrm{min}}\limits_{\FF} &~~
	\fpeb(\FF; \etabtilde) ~~
	\mathrm{s.t.} ~~ \tracesmall{\FF \FF^H} = \revv{\Ptot/(KL)} ~.
\end{align}
Using change of variables $\XXb = \revv{L} \FF \FF^H$, \eqref{eq_problem_peb} can be relaxed to (by removing the constraint ${\rm{rank}}(\XXb) = M$) \cite[Ch.~7.5.2]{boyd2004convex}
\begin{align} \label{eq_problem_peb2}
	\mathop{\mathrm{min}}\limits_{\XXb, \uu} &~~
	\boldone^T \uu 
	\\ \nonumber 
    \mathrm{s.t.}&~~\begin{bmatrix}  \JJpos\rev{(\XXb)} & \ekk{b} \\ \ekkt{b} & u_b  \end{bmatrix} \succeq 0, ~ b=0,1, 
	~ \tracesmall{\XXb} = \revv{\frac{\Ptot}{K}},  \XXb \succeq 0 ~,
\end{align}
where $\uu = \left[u_0, u_1 \right]^T$ is a newly introduced auxiliary variable and $\ekk{b}$ is the $\thnew{b}$ column of the identity matrix. It can be observed from \rev{Appendix~\ref{app_fim_der}} and \eqref{eq_jjpos} that $\JJpos\rev{(\XXb)}$ is a linear function of $\XXb$ \cite{li2007range,precoderNil2018}, which implies that \eqref{eq_problem_peb2} is a convex semidefinite program (SDP) \cite[Ex.~3.26]{boyd2004convex} and thus can be solved via standard convex optimization tools.\footnote{Using the covariance matrix $\XXb \in \complexset{\Ntx}{\Ntx}$ obtained from the relaxed problem in \eqref{eq_problem_peb2}, an approximate precoder $\FF \in \complexset{\Ntx}{M}$ can be recovered by using randomization procedures \cite{luo2010semidefinite}.} 
The following proposition provides the structure of the precoder covariance matrix \rev{for the relaxed problem in \eqref{eq_problem_peb2}}.\vspace{-0.05in}

\begin{proposition}\label{prop_opt_prec}
    The precoder covariance matrix $\XXbstar$ obtained as the solution to \rev{the relaxed problem in} \eqref{eq_problem_peb2} can be expressed as
    \begin{equation}
        \XXbstar = \UU \Lambdab \UU^H ~,
    \end{equation}
    where $\Lambdab \in \complexset{2\revv{G}}{2\revv{G}}$ is a positive semidefinite (PSD) matrix and $ \UU \triangleq [ \AAtxthetai ~ \AAtxthetadtalli  ]^{*}$, $\AAtxthetai \triangleq [ \atx(\theta_0) \, \ldots \,  \atx(\theta_{\revv{G}-1}) ]$, $\AAtxthetadtalli \triangleq [ \atxdt(\theta_0) \, \ldots \,  \atxdt(\theta_{\revv{G}-1}) ]$, 
    with $\atxdt(\theta) \triangleq {\partial \atx(\theta)}/{\partial \theta}$.
\end{proposition}
\begin{proof}
    See Appendix~\ref{app_prop1}.\vspace{-0.05in}
\end{proof}

Proposition~\ref{prop_opt_prec} can significantly reduce the computational burden of solving \eqref{eq_problem_peb2} since the optimization can equivalently be performed over $\Lambdab \in \complexset{2\revv{G}}{2\revv{G}}$ instead of $\XXb \in \complexset{\Ntx}{\Ntx}$. Note that due to channel sparsity in mmWave \cite{el2014spatially}, we usually have $\revv{G} \ll \Ntx$. Hence, we propose to solve the following equivalent convex problem:
\begin{align} \label{eq_problem_peb3}
	\mathop{\mathrm{min}}\limits_{\Lambdab, \uu} ~
	\boldone^T \uu \quad 
    \mathrm{s.t.}&\begin{bmatrix} \JJpos\rev{(\Lambdab)} & \ekk{b} \\ \ekkt{b} & u_b  \end{bmatrix} \succeq 0,  b=0,1,\\
	&\tracesmall{\UU \Lambdab \UU^H} = \revv{\frac{\Ptot}{K}},	 \Lambdab \succeq 0.\nonumber 
\end{align}
Based on Proposition~\ref{prop_opt_prec}, it is worth highlighting the remarkable connection between mmWave positioning and \textit{monopulse radar} \cite{monopulse_review}. Precisely, $\AAtxthetai$ and $\AAtxthetadtalli$ correspond to \textit{sum} and \textit{difference} beams employed in monopulse radar for target tracking with high precision. In the context of mmWave signal design, UE can be tracked by using a combination of \textit{directional} and \textit{derivative} beams pointing towards the path AODs.\footnote{Note from Appendix~\ref{app_prop1} that each beam in $\FF$ is a weighted combination of the directional and derivative beams in $\UU$.} \rev{The key difference is that a mmWave BS can exploit multipath to enable positioning under imperfect BS-UE synchronization while monopulse radar relies on LOS-only tracking with narrow beamwidth (i.e., $\AAtxthetai$ and $\AAtxthetadtalli$ contain a single beam in a monopulse system \cite{monopulse_review}). Fig.~\ref{fig_sum_diff_beampattern} illustrates the beampatterns of directional and derivative beams, which agree with those of sum and difference beams in \cite[Fig.~2]{monopulse_review}.}



\vspace{-2mm}
\subsection{Robust Signal Design with Imperfect Knowledge}\label{sec_robust} \vspace{-0.05in}
Assuming that $\etabtilde$ belongs to an uncertainty region $\mathcal{U}$ \cite{precoderNil2018}, we propose two strategies for robust signal design. In practice, $\mathcal{U}$ can be determined from the output of tracking routines, where the means and covariances of the parameters specify, respectively, the center and the extent of $\mathcal{U}$ \cite{Mendrzik_JSTSP_2019}. 

\subsubsection{Optimization-Based Robust Unconstrained Design}\label{sec_opt_robust}
Under imperfect knowledge of $\etabtilde$, we resort to the worst-case PEB minimization strategy:
\begin{align} \label{eq_problem_worst_case}
	\mathop{\mathrm{min}}\limits_{\XXb} &~~	\mathop{\mathrm{max}}\limits_{\etabtilde \in \mathcal{U}} ~
	\fpeb(\XXb; \etabtilde)
	~~	\mathrm{s.t.}~~ \tracesmall{\XXb} = \revv{\frac{\Ptot}{K}} ,~  \XXb \succeq 0 ~.
\end{align}
Discretizing $\mathcal{U}$ into \rev{a uniform grid of} $N$ points $\{ \etabtilde_n \}_{n=0}^{N-1}$, the problem in \eqref{eq_problem_worst_case} can be reformulated in the epigraph form as
\begin{align} \label{eq_problem_worst_case_convex2}
	&\mathop{\mathrm{min}}\limits_{\XXb,t,\{u_{n,b}\}} ~ t	
	\\ \nonumber 
	&\mathrm{s.t.}~\begin{bmatrix}  \JJpos(\rev{\XXb}; \etabtilde_n) & \ekk{b} \\ \ekkt{b} & u_{n,b}  \end{bmatrix} \succeq 0, ~ b=0,1, ~ n = 0, \ldots, N-1
	\\ \nonumber &\hspace{7mm} u_{n,0} +  u_{n,1} \leq t, ~n=0,\ldots,N-1 \\ \nonumber
	&\hspace{7mm} \tracesmall{\XXb} = \revv{\Ptot/K} ,~  \XXb \succeq 0, 
\end{align}
where $\JJpos(\rev{\XXb;}\etabtilde_n)$ represents the FIM in \eqref{eq_jjpos} evaluated at $\etabtilde = \etabtilde_n$, and $t$ and $\{u_{n,b}\}_{\forall n,b}$ are auxiliary variables. Similar to \eqref{eq_problem_peb2}, the problem in \eqref{eq_problem_worst_case_convex2} is an SDP and can be solved using convex optimization \cite{cvx}. 


\begin{figure}
	\centering
	\includegraphics[width=0.75\linewidth]{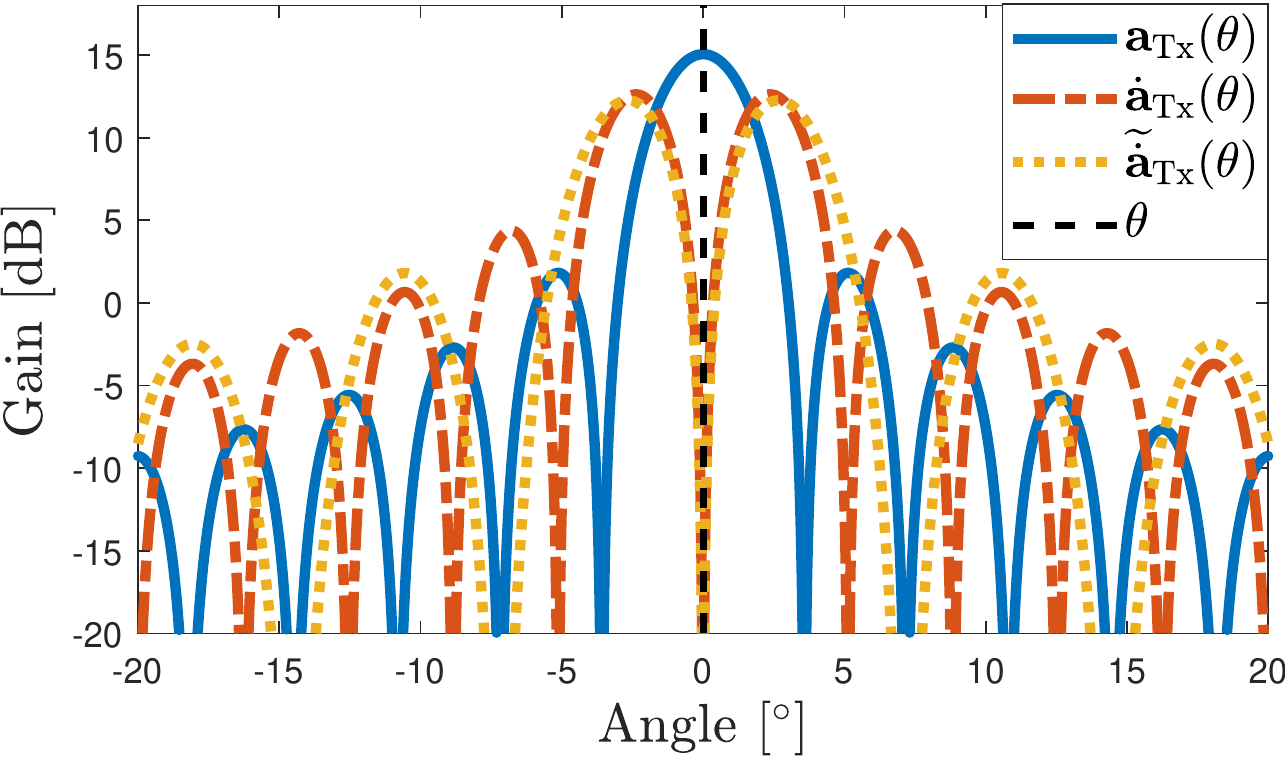}
	\caption{Uniform linear array (ULA) beampatterns of the directional beam $\atx(\theta)$, the digital derivative beam $\protect\atxdt(\theta)$ and the analog derivative beam $\protect\atxdttilde(\theta)$ for $\theta = 0$ with $\Ntx=32$ antennas. Intuitively, the directional beam provides the required SNR for position estimation, while the derivative beam helps the UE to best identify small deviations from the boresight direction, like in monopulse radar, as evidenced by its sharp curvature around the boresight. High-precision tracking of UE involves weighted combination of directional and derivative beams.}
	\label{fig_sum_diff_beampattern}
	\vspace{-0.1in}
\end{figure}

\subsubsection{Codebook-Based Heuristic Design with Optimized Beam Power Allocation}\label{sec_codebook_design}
Due to the presence of multiple grid points $\{ \etabtilde_n \}_{n=0}^{N-1}$, \eqref{eq_problem_worst_case_convex2} cannot be transformed into a lower dimensional problem, as in \eqref{eq_problem_peb3}. To devise a low-complexity signal design approach as an alternative to the SDP in \eqref{eq_problem_worst_case_convex2}, we propose a codebook-based heuristic design inspired by Proposition~\ref{prop_opt_prec}. Specifically, we consider the following digital and analog codebooks to span the AOD uncertainty intervals of the different paths:
\begin{align} \label{eq_ffdig}\vspace{-0.2in}
    \FFdigital \triangleq \left[ \FFsum ~ \FFdiff \right] \,,~ 
    \FFanalog \triangleq \left[ \FFsum ~ \FFdifftilde \right] ~,
\end{align}
where 
   $\FFsum \triangleq \left[ \FFsum_0 \, \ldots \, \FFsum_{{\revv{G}}-1} \right]$, 
    $\FFdiff \triangleq \left[ \FFdiff_0 \, \ldots \, \FFdiff_{{\revv{G}}-1} \right]$, 
    $\FFdifftilde \triangleq \left[ \FFdifftilde_0 \, \ldots \, \FFdifftilde_{{\revv{G}}-1} \right]$, in which  
\begin{align}\label{eq_codebook_def}\vspace{-0.2in}
    \FFsum_{\revv{g}} &\triangleq \left[ \atx(\theta_{{\revv{g}},0}) \, \ldots \, \atx(\theta_{{\revv{g}},N_{\revv{g}}-1}) \right] ~, \\
    \FFdiff_{\revv{g}} &\triangleq \left[ \atxdt(\theta_{{\revv{g}},0}) \, \ldots \, \atxdt(\theta_{{\revv{g}},N_{\revv{g}}-1}) \right] ~,  \\
    \FFdifftilde_{\revv{g}} &\triangleq \left[ \atxdttilde(\theta_{{\revv{g}},0}) \, \ldots \, \atxdttilde(\theta_{{\revv{g}},N_{\revv{g}}-1}) \right] ~,
\end{align}
for ${\revv{g}} \in \{0, \ldots, {\revv{G}}-1\}$. Here, $\{ \theta_{{\revv{g}},i} \}_{i=0}^{N_{\revv{g}}-1}$ denote the evenly spaced AODs covering the uncertainty interval of the $\thnew{{\revv{g}}}$ path, with an angular spacing equal to $3 \, \rm{dB}$ (half-power) beamwidth \cite{zhang2018multibeam}, \cite[Ch.~22.10]{orfanidis2002electromagnetic}. Moreover, $\atxdttilde(\theta)$ denotes the best analog approximation (i.e., with unit-modulus entries) 
to $\atxdt(\theta)$ determined using gradient projections iterations in \cite[Alg.~1]{analogBeamformerDesign_TSP_2017}. In \eqref{eq_ffdig}, $\FFsum$ represents a standard \textit{directional} beam codebook
, while $\FFdiff$ and $\FFdifftilde$ are novel \textit{derivative} codebooks stemming from Proposition~\ref{prop_opt_prec}.  \rev{Note that in the regime of large number of antennas, broader beams than steering vectors should be used to avoid very large $N_{\revv{g}}$.}


\begin{table}
\caption{Simulation Parameters}
\centering
\scriptsize{
    \begin{tabular}{|l|l||l|l|}
        \hline
        $\fc$ & $28 \, \rm{GHz}$ & BS loc. $\qq$ & $\left[0, 0 \right] ~ \rm{m}$  \\ \hline   
        $K$ & $1024$ & UE loc. $\pp$ & $\left[25, 10 \right] ~ \rm{m}$ \\ \hline 
        $\deltaf$ & $120 \, \rm{kHz}$ & \revv{incidence} loc. $\rr_1$ & $\left[15, 25 \right] ~ \rm{m}$  \\ \hline
        $L$ & $1$ & noise PSD $N_0$ & $-174 ~ \rm{dBm/Hz}$     \\ \hline
        $\Ntx$ & $32$ &  noise figure $\Frx$ & $8 ~ \rm{dB}$   \\ \hline 
        $\Nrx$ & $16$ & $\sigma^2$ & $10^{0.1 (\Frx + N_0)} K \deltaf $    \\ \hline
        $\Nrxrf$ & $16$ &  Tx Power, \revv{$\Ptot/(LM)$} & $20 ~ \rm{dBm}$   \\ \hline 
        \rev{$M$ (Scen.1)} & \rev{$16$} &  \rev{$M$ (Scen.2)} & \rev{$8$} \\ \hline
        \rev{$N$ (Scen.1)} & \rev{$36$} &  \rev{$N$ (Scen.2)} & \rev{$16$} \\ \hline
    \end{tabular}}
    \label{tab_parameters}
    \vspace{-0.1in}
\end{table}

Given a predefined codebook with entries $\FF = \left[ \fb_0 \, \ldots \, \fb_{M-1} \right]$ (with $M=2 \sum_{\revv{g}=0}^{\revv{G}-1} N_{\revv{g}}$), \revv{each normalized to have the squared norm $\norm{\fb_m}_2^2 = \Ptot / (KLM) ~ \forall m$}, we consider the optimal beam power allocation problem in \rev{$\rhob=[\rho_0,\ldots,\rho_{M-1}]^T$}:
\begin{align} \label{eq_problem_beampower}
	\mathop{\mathrm{min}}\limits_{\rev{\rhob},t,\{u_{n,k}\}} &~~ t	
	\\ \nonumber 
	\mathrm{s.t.}&~~\begin{bmatrix}  \JJpos(\rev{\XXb}; \etabtilde_n) & \ekk{b} \\ \ekkt{b} & u_{n,b}  \end{bmatrix} \succeq 0, ~ b=0,1, ~ n = 0, \ldots, N-1 
	\\ \nonumber &~~ u_{n,0} +  u_{n,1} \leq t, ~n=0,\ldots,N-1  \\ \nonumber
	&~~ \XXb = \revv{L} \FF \, \diag{\rev{\rhob}}  \FF^H, ~ \tracesmall{\XXb} = \revv{\Ptot/K} ,~  \rev{\rhob} \geq \mathbf{0} ~,
\end{align}
which yields the optimized codebook $\FFopt = \left[ \rev{\sqrt{\rho_0}} \fb_0 \, \ldots \, \rev{\sqrt{\rho_{M-1}}} \fb_{M-1} \right]$. 


\subsubsection{\revv{Time Sharing Optimization}}\label{sec_timeshare_alg}
\revv{When the BS transmits at a maximum power $\Pmax$ per symbol, the power allocation formulation in \eqref{eq_problem_beampower} can be equivalently used to optimize the beam \textit{time sharing factors} $L_m, \forall m$ (i.e., the number of times $\fb_m$ is transmitted with a fixed maximum power) since the expression $\XXb = \revv{L} \FF \, \diag{{\rhob}}  \FF^H$ is valid if $L_m = L \rho_m$. To cover such scenarios, we provide a signal design algorithm for time sharing implementation in Algorithm~\ref{alg_time_main}, where continuous power values are mapped to discrete time sharing factors.}

\begin{algorithm}
	\caption{\revv{Time Sharing Optimization via Power Allocation}}
	\label{alg_time_main}
	\begin{algorithmic}[1]
	    \State \revv{\textbf{Input:} Number of symbols per beam $L$, maximum transmit power per symbol $\Pmax$.}
	    \State \revv{\textbf{Output:} Time sharing factors $L_0, \, \ldots , \, L_{M-1}$ }
	    \begin{enumerate}
	        \item \revv{Solve the power allocation problem in \eqref{eq_problem_beampower} with the total power constraint $\Ptot = L M \Pmax $, yielding $\rho_0 ,\, \ldots ,\, \rho_{M-1}$.}
	        \item \revv{Set the time sharing factors as $L_m = \nint{L \rho_m}, ~ m = 0, \ldots, M-1$ by rounding to the nearest integer.}
	        \end{enumerate}
	\end{algorithmic}
\end{algorithm}

\subsubsection{\rev{Complexity Analysis}}\label{sec_complexity}
\rev{In this part, we analyze the complexity of the SDPs in \eqref{eq_problem_worst_case_convex2} and \eqref{eq_problem_beampower}. According to \cite[Ch.~11]{nemirovski2004interior}, the complexity of an SDP is given by $\mathcal{O}\left( n^2 \sum_{i=1}^{C} m_i^2 + n \sum_{i=1}^{C} m_i^3 \right)$, where $n$ is the number of optimization variables, $C$ is the number of linear matrix inequality (LMI) constraints, and $m_i$ is the row/column size of the matrix associated with the $\thnew{i}$ LMI constraint. In \eqref{eq_problem_worst_case_convex2}, we have $n = \Ntx^2 + 2N + 1$, $C = 2N+1$, $m_i = 4\revv{G}+3$ for $i < C$ and $m_i = \Ntx$ for $i = C$. Assuming $N \ll \Ntx^2$ (which holds in practice, as seen from Table~\ref{tab_parameters}) and $\revv{G} \ll \Ntx$ (due to channel sparsity), this yields an approximate complexity of $\mathcal{O}(\Ntx^6)$.}
\rev{Following similar arguments, the complexity of the SDP in \eqref{eq_problem_beampower} is roughly given by $\mathcal{O}(N^3)$ under the conditions $N \approx M$ (since both parameters depend on the size of $\mathcal{U}$) and $\revv{G} \ll N$.}
Clearly, \eqref{eq_problem_beampower} offers a much more efficient approach to robust signal design than \eqref{eq_problem_worst_case_convex2} as \rev{$N \ll \Ntx^2$} in practical scenarios (which will be verified in the next section through execution time analysis). \rev{Moreover, the complexity of \eqref{eq_problem_worst_case_convex2} is comparable to that of the unconstrained design in \cite[Eq.~(23)]{precoderNil2018}, while \eqref{eq_problem_beampower} incurs a similar level of complexity to the codebook-based design in \cite[Eq.~(46)]{tasos_precoding2020}.}

\vspace{-0.05in}
\section{Numerical Results}\vspace{-0.05in}\label{sec_num_res}
To evaluate the performance of the proposed signal design algorithms, simulations are carried out using the parameters in Table~\ref{tab_parameters} \cite{kakkavas2019performance,tasos_precoding2020}, where the BS has a uniform linear array (ULA) and the UE is equipped with a uniform circular array (UCA) so that the PEB does not depend on its unknown orientation\footnote{The squared array aperture function (SAAF) \cite[Def.~2]{Array_Loc_Win_2016}, which provides complete characterization of the effect of array geometry on the PEB, is independent of the AOA of the incident signal for UCAs \cite[Def.~4]{Array_Loc_Win_2016}, meaning that the PEB is independent of the orientation. However, the UE should still estimate the orientation as an unknown nuisance parameter for positioning.} $\psi$ \cite{Array_Loc_Win_2016}. To collect energy from all directions, the receive combiner is set as $\WW = \Imatrix_{\Nrxrf}$ \cite{kakkavas2019performance}, which can equivalently be realized by an analog array at the UE with a DFT codebook over $\Nrx\times M$ frames. We consider a two-path environment ($\revv{G}=2$) with a LOS and an NLOS path\footnote{Due to space limitations, results with higher number of NLOS paths are not presented. However, the main trends and conclusions will still be valid, with the difference pertaining to the saturated PEB values \cite{kakkavas2019performance} in the NLOS-limited regime (which will be described in Fig.~\ref{fig_peb_regimes}).}, whose gains are given, respectively, by $\abs{\alpha_0} = c / (4 \pi \fc \norm{ \pp - \qq }_2 )$ and $\abs{\alpha_1} = \gamma \, c / (4 \pi \fc [\norm{ \qq - \rr_1 }_2 + \norm{ \rr_1 - \pp }_2] )$, where $\gamma$ denotes the NLOS path reflection coefficient. The phases of the complex gains are assumed to be uniformly distributed over $\left[-\pi, \pi \right]$. Moreover, we explore two scenarios where the uncertainty in the \revv{incidence point} location in $x$ and $y$ directions is taken to be $5 \, \rm{m}$ and $0.3 \, \rm{m}$ in Scenario~1 and Scenario~2, respectively, while the uncertainty in UE location is fixed to $0.3 \, \rm{m}$ in both scenarios\revv{\footnote{\label{fn_uncert}\revv{ $N=36$ (resp. $16$) points are constructed by combining $9$ (resp. $4$) incidence and $4$ UE locations, all uniformly distributed in the corresponding uncertainty regions. As a general guideline, $N$ could be chosen such that the angular separation between the grid points in the $x-y$ plane is close to the corresponding beamwidth, which allows the locations in between the grid points to be covered by the same set of beams.}}}. We model the clock bias as \rev{$p(\deltat)=\mathcal{N}(\deltat; 0,\sigmaclk)$}
, with $\JJclk={1}/{\sigmaclk}$ \cite{tasos_precoding2020}. 
\rev{To establish a clear connection between the clock bias uncertainty $\sigmaclkk$ and the PEB, $\sigmaclkk$ will be expressed in meters instead of seconds.}

\vspace{-0.05in}
\subsection{Analysis of PEB Regimes against Clock Bias Uncertainty}\label{sec_peb_regimes}
To investigate the positioning performance under different 
$\sigmaclkk$, Fig.~\ref{fig_peb_regimes} plots the PEBs in Scenario~1 achieved by the digital codebook $\FFdigital$ in \eqref{eq_ffdig}, optimized using \eqref{eq_problem_beampower}, for various values of the NLOS reflection coefficient $\gamma$, along with the PEB curve for $\gamma = 0.1$ obtained by assuming perfect AOD and TOA estimation for the LOS path.\footnote{This PEB analysis is needed to better understand the results obtained with the different signal designs in Sec.~\ref{sec_perf}.} We observe that for small $\sigmaclkk$ (i.e., almost perfect synchronization), the PEB is mainly limited by the accuracy of the LOS parameters as the position information can be obtained by utilizing only the LOS path. 
In this \textit{LOS-limited regime}, perfect AOD and TOA information for the LOS path improves PEB up to its limit determined by $\sigmaclkk$, while a larger NLOS gain $\gamma$ does not have any effect on PEB. 
On the other hand, for high $\sigmaclkk$, \rev{which is more relevant in practice considering 5G specifications \cite[Table~6.1]{TR38.857} (where $\sigmaclkk = 15 \, \rm{m}$ is used)}, a larger $\gamma$ improves PEB performance significantly. This is because the LOS path alone cannot provide sufficient information for positioning due to large clock bias uncertainty and the NLOS path contributes to synchronizing the UE clock (and, thus positioning) in this \textit{NLOS-limited regime} (where even perfectly estimated LOS parameters do not improve positioning accuracy). In the absence of an NLOS path (i.e., $\gamma = 0$), the PEB grows unbounded with $\sigmaclkk$ as the TOA of the LOS path carries almost no ranging information for high $\sigmaclkk$ and the AOD information is not enough to locate the UE. 

Finally, we observe a transition region with unit slope that connects the LOS- and NLOS-limited regimes. In this \textit{clock bias prior-limited regime}, perfect LOS information does not improve PEB because $\sigmaclkk$ is large enough to dominate the position estimation error, while NLOS information does not help positioning either because $\sigmaclkk$ is small enough such that the prior information on clock bias becomes much more significant than the observation-related information coming from NLOS TOA and AOD estimates. 
After a certain level of $\sigmaclkk$ (which depends on $\gamma$), the effect of prior information becomes negligible and the observation-related information becomes dominant, which yields the saturated curves in the NLOS-limited regime. 

\begin{figure}
	\centering
	\includegraphics[width=0.8\linewidth]{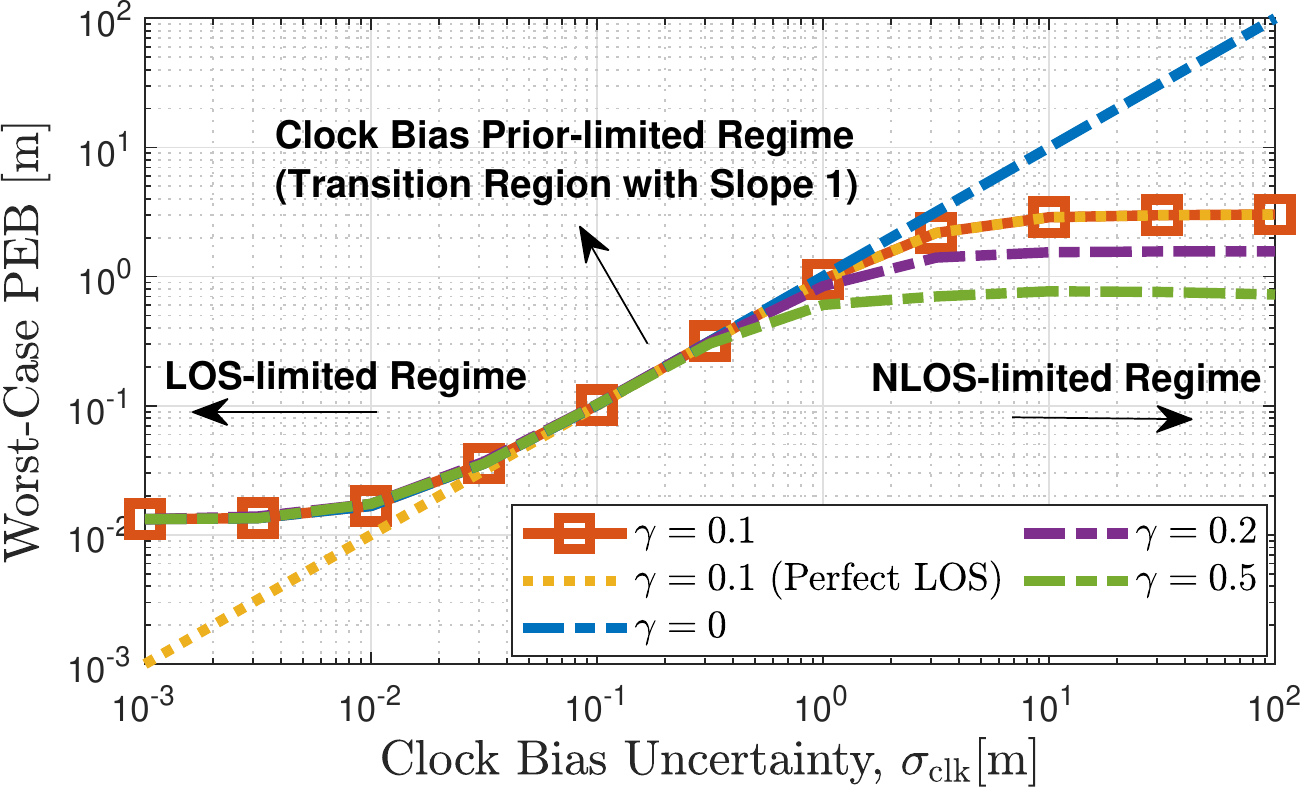}
	\caption{PEB regimes with respect to clock bias uncertainty for the digital codebook $\FFdigital$, where $\gamma$ is the reflection coefficient of the NLOS path.}
	\label{fig_peb_regimes}
	\vspace{-0.15in}
\end{figure}


\subsection{Performance of Optimal and Codebook Based Strategies}\label{sec_perf}
In this part, we compare the worst-case PEB performances of the following signal design strategies by setting $\gamma = 0.1$:
\begin{itemize}
    \item \textit{Proposed:} Optimal robust design in \eqref{eq_problem_worst_case_convex2}, digital and analog codebooks, $\FFdigital$ and $\FFanalog$ in \eqref{eq_ffdig}, optimized via \eqref{eq_problem_beampower}. 
    \item \textit{Benchmark:} Directional beam codebook $\FFsum$ \cite{tasos_precoding2020}, optimized via \eqref{eq_problem_beampower} and with uniform power allocation.
\end{itemize}
Fig.~\ref{fig_ex1_peb_power} shows the worst-case PEBs and the relative illuminations of the LOS path\footnote{For a given precoder $\FF$, the relative LOS illumination is calculated as $\int_{\theta \in {\rm{LOS}}} \norm{\atx^T(\theta) \FF}_2^2 \mathrm{d}\theta/ \int_{\theta \in \{{\rm{LOS}},{\rm{NLOS}}\}} \norm{\atx^T(\theta) \FF}_2^2\mathrm{d}\theta$. As the relative LOS illumination exhibits similar trends for Scenario~1 and Scenario~2, the results are shown only for Scenario~1.} for Scenario~1 and Scenario~2. Moreover, Fig.~\ref{fig_beampattern} illustrates the beampatterns of the considered strategies in Scenario~1 for two different values of $\sigmaclkk$. Finally, the average execution times on MATLAB for Scenario 1 were found to be 17.3 s (to solve \eqref{eq_problem_worst_case_convex2}), 3.8 s (using $\FFdigital$ and $\FFanalog$), 3.7 s (using $\FFsum$).


\begin{figure}
	\centering
	\includegraphics[width=0.8\linewidth]{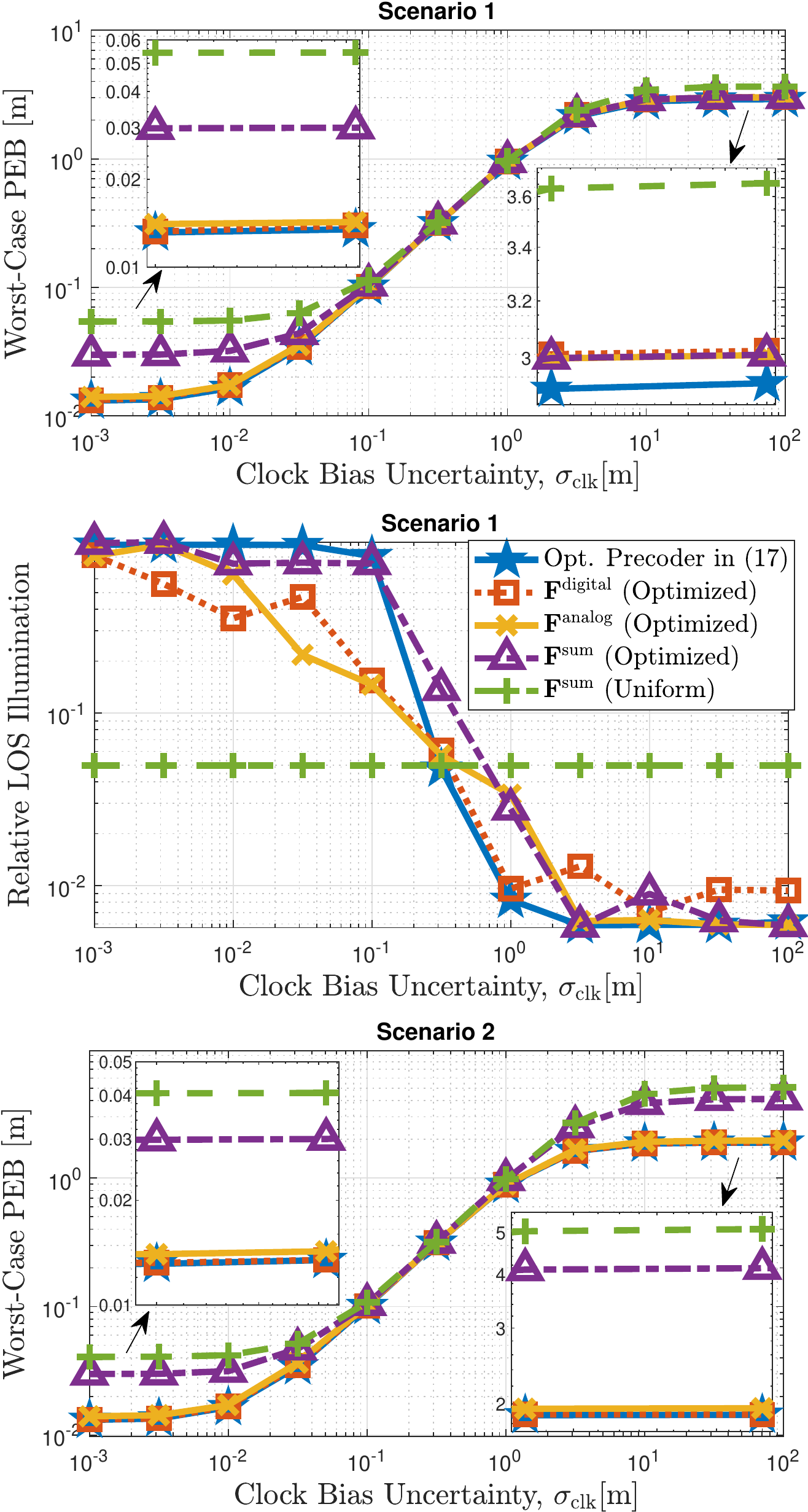}
	\caption{Worst-case PEBs and relative illuminations of the LOS path for the considered signal design strategies in Scenario~1 (large \revv{incidence} location uncertainty) and Scenario~2 (small \revv{incidence} location uncertainty). In Scenario~1, $N_0 = 2$ and $N_1 = 6$ beams are needed to cover the uncertainty region of the UE and the \revv{incidence point}, respectively, while we have $N_0 = N_1 = 2$ in Scenario~2. \revv{As observed from the uniform and optimized $\FFsum$ curves, the proposed power allocation in \eqref{eq_problem_beampower} always improves positioning performance.}}
	\label{fig_ex1_peb_power}
	\vspace{-0.15in}
\end{figure}

In agreement with the observations in Fig.~\ref{fig_peb_regimes}, all strategies experience the three different PEB regimes with respect to $\sigmaclkk$. As expected, most of the power budget is allocated to illuminate the LOS path in the LOS-limited regime, whereas the NLOS path gets more power as $\sigmaclkk$ increases. Additionally, the proposed optimal strategy in \eqref{eq_problem_worst_case_convex2} achieves the best performance in all operation regimes and scenarios. In Scenario~1, the proposed codebooks $\FFdigital$ and $\FFanalog$ outperform the traditional (optimized) codebook $\FFsum$ in the LOS-limited regime, while they achieve almost the same performance in the NLOS-limited regime. The reason is that although directional beams maximize the SNR over the specified angular region (see Fig.~\ref{fig_beampattern}), positioning requires accurate AOD estimation, which can be enabled by a balanced combination of directional and derivative beams, as illustrated in Fig.~\ref{fig_sum_diff_beampattern}. 
Due to the relatively small uncertainty in the LOS path compared to the NLOS path in Scenario~1, $\FFdigital$ and $\FFanalog$ can provide non-negligible performance gains over $\FFsum$ in the LOS-limited regime, whereas $\FFsum$ in the NLOS-limited regime can produce beampatterns very similar to those achieved by $\FFdigital$ and $\FFanalog$ over a large angular region, as seen from Fig.~\ref{fig_beampattern}. This is also corroborated by the PEB results of Scenario~2, where $\FFdigital$ and $\FFanalog$ significantly outperform $\FFsum$ in both LOS- and NLOS-limited regimes due to low uncertainty in UE and \revv{incidence point} locations. 

Comparing the proposed optimal and codebook based strategies, we see that $\FFdigital$ and $\FFanalog$ exhibit PEB performances very close to the optimal approach in \eqref{eq_problem_worst_case_convex2} in most cases, which implies that the proposed codebooks with optimized power allocation provide near-optimal and low-complexity solutions to spatial signal design. In addition, for all regimes and scenarios, the performance gap between $\FFdigital$ and $\FFanalog$ is negligible, which suggests that the proposed methods can be implemented in low-cost analog mmWave architectures 
without any loss in positioning accuracy. To summarize the main results, Table~\ref{tab_guide} provides practical guidelines on which spatial signals to use under different PEB regimes.
\vspace{-0.1in}



\begin{figure}
	\centering
	\includegraphics[width=0.8\linewidth]{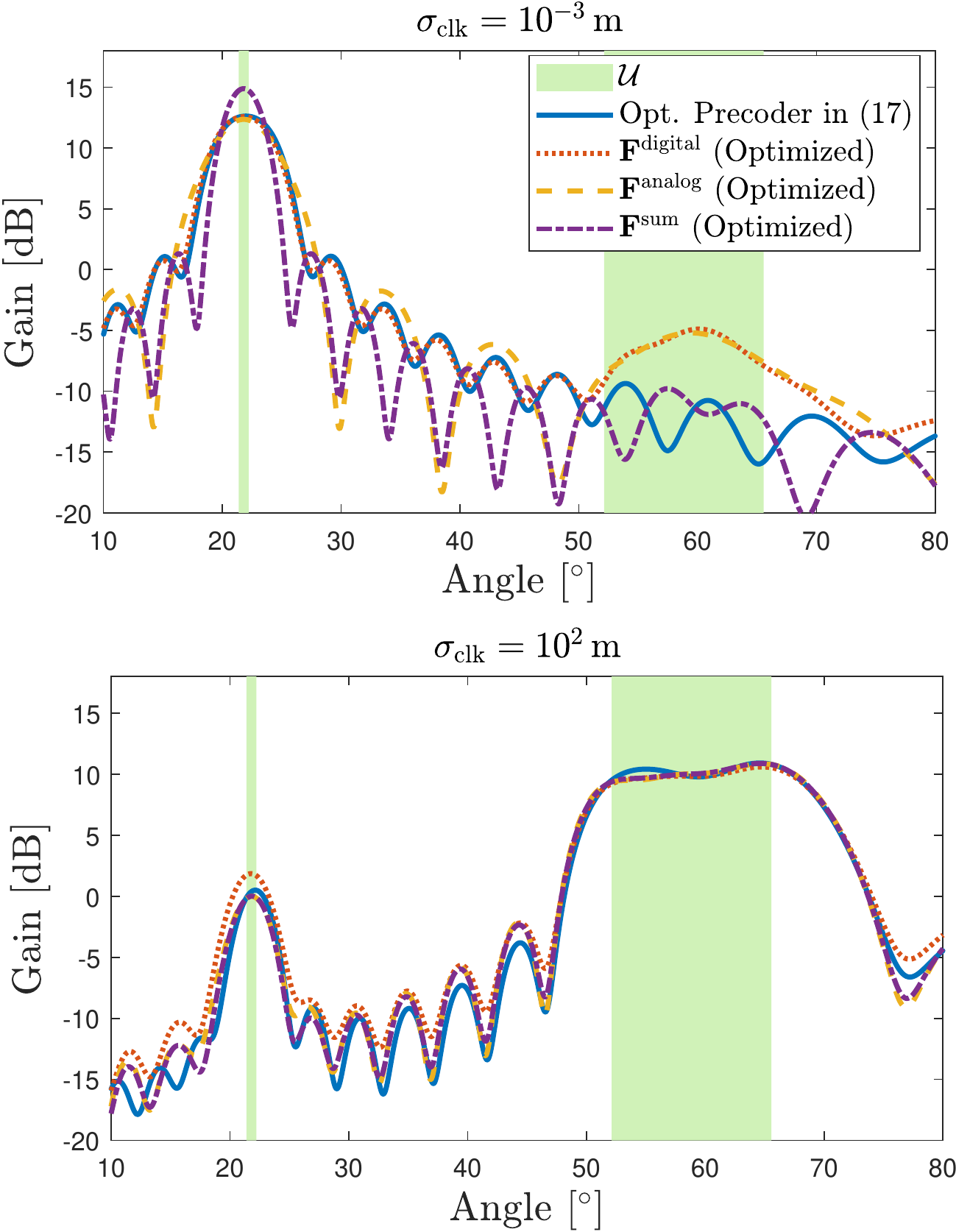}
	\caption{Aggregated beampatterns ($\norm{\atx^T(\theta) \FF}_2^2$) of the different strategies in Scenario~1 for two different levels of \revv{$\sigmaclkk$}. While $\FFsum$ maximizes the SNR towards the UE location for low $\sigmaclkk$, it is not \revv{positioning-optimal} as the derivative beams in $\FFdigital$ and $\FFanalog$ \revv{are needed for optimality}.}
	\label{fig_beampattern}
	\vspace{-0.15in}
\end{figure}


\begin{table}
\caption{Guidelines on Which Spatial Signal to Use under Different PEB Regimes in Fig.~\ref{fig_peb_regimes}\vspace{-0.1in}}
\centering
\scriptsize{
    \begin{tabular}{|l|l|l|l|}
        \hline
        & \textbf{LOS-limited} & \textbf{Clock bias prior-limited} & \textbf{NLOS-limited}   \\ \hline
        \textbf{Scenario~1} & $\FFanalog$ & $\FFsum$ (Uniform) & Opt. Precoder in \eqref{eq_problem_worst_case_convex2}  \\ \hline
        \textbf{Scenario~2} & $\FFanalog$ & $\FFsum$ (Uniform) & $\FFanalog$ \\ \hline
    \end{tabular}}
    \label{tab_guide}
    \vspace{-0.1in}
\end{table}

\subsection{\revv{Performance of Time Sharing Optimization}}\label{sec_timeshare}
\revv{To validate Algorithm~\ref{alg_time_main}, we plot in Fig.~\ref{fig_digital_L_timeshare} the worst-case PEBs obtained by the power allocation scheme in \eqref{eq_problem_beampower} and the time sharing scheme defined by Algorithm~\ref{alg_time_main}. As expected, with sufficient number of symbols per beam, time sharing can attain the same performance as power allocation since the rounding error in $L_m = \nint{L \rho_m}$ diminishes as $L$ increases.}

\begin{figure}
	\centering
	\includegraphics[width=0.8\linewidth]{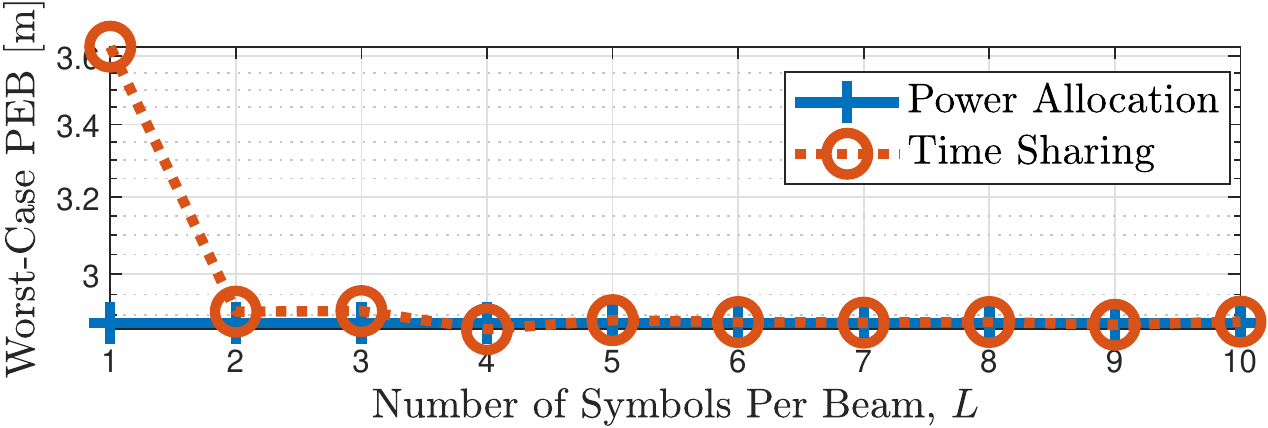}\vspace{-0.15in}
	\caption{\revv{Power allocation via \eqref{eq_problem_beampower} versus time sharing optimization via Algorithm~\ref{alg_time_main} with respect to $L$ at $\sigmaclkk = 10 \, \rm{m}$ for $\FFdigital$ in Scenario~1.}}
	\label{fig_digital_L_timeshare}
	\vspace{-0.25in}
\end{figure}

\vspace{-0.05in}
\section{Conclusion}\vspace{-0.05in}
We have tackled the problem of spatial signal design for mmWave positioning under clock bias and location uncertainties. The proposed optimization-based robust design and low-complexity codebook-based heuristic designs are shown to provide significant improvements in positioning accuracy over traditional methods. Since the optimal precoder structure is valid for any scalar function of the FIM, the proposed approach can also be applied to other objectives than PEB, such as orientation error bound (OEB). Remarkably, the analog codebook can attain almost the same PEB performance as the digital one \rev{(assuming the usage of single precoding vector at the BS for each frame)}, which proves that high-accuracy positioning can be enabled by analog-only beamforming architectures. 
This result is also promising for beyond 5G systems where reconfigurable intelligent surfaces (RISs) with analog phase shifters are expected to be key enablers for positioning. \rev{As future research, we plan to investigate extensions to multi-user scenarios under multi-user or multi-BS interference constraints and frequency-dependent precoding structures (i.e., $\fb_m$ will depend on the subcarrier index $k$).}




\vspace{-0.1in}
\begin{appendices}
\vspace{-0.03in}
\section{\rev{FIM Elements as a Function of $\FF$ and $\XXb$}}\label{app_fim_der}
\revv{Assuming QPSK pilots, i.e., $\abs{\skl}^2 = 1$, } \rev{the FIM elements in \eqref{eq_jobs_ij} can be derived as follows:
\begin{align}
     \left[ \JJ \right]_{i,j} 
    & = \frac{2}{\sigma^2}  
    \sum_{k,\ell,m}
    \realp{ \skl^{*} \fb_m^H   \frac{\partial \HH_k^H}{\partial \eta_i} \WW   \WW^H  \frac{\partial \HH_k}{\partial \eta_j} \fb_m \skl }   \notag
     \\ \nonumber
    &= \frac{2}{\sigma^2}  
    \sum_{k,\ell,m}
    \realp{ \abs{\skl}^2 \tracesmall{ \fb_m \fb_m^H   \frac{\partial \HH_k^H}{\partial \eta_i} \WW   \WW^H  \frac{\partial \HH_k}{\partial \eta_j}   } } 
     \\ 
    &= \frac{2}{\sigma^2}  \sum_{k=0}^{K-1}  \realp{  \tracesmall{  \XXb \frac{\partial \HH_k^H}{\partial \eta_i} \WW   \WW^H  \frac{\partial \HH_k}{\partial \eta_j}    } } ~, \label{eq_J_open}
\end{align}
where $\XXb \triangleq \revv{L} \FF \FF^H = \revv{L} \sum_{m=0}^{M-1} \fb_m \fb_m^H$. It follows from \eqref{eq_J_open} that the channel domain FIM $\JJ$ in \eqref{eq_jobs_ij} and the position domain FIM $\JJpos$ in \eqref{eq_jjpos} are functions of $\FF$ and $\XXb$.}

\vspace{-0.1in}
\section{Proof of Proposition~\ref{prop_opt_prec}}\label{app_prop1}
\vspace{-0.07in}
    First, we observe from the structure of $\JJ$ in \eqref{eq_jobs_ij} that the dependence of $\JJ$ on $\XXb$ is only through the following matrices:
    \begin{equation}
        \AAtxthetaih \XXb \AAtxthetaiconj, ~ \AAtxthetadtallih \XXb \AAtxthetaiconj, ~ \AAtxthetadtallih \XXb \AAtxthetadtalli^{*}~.
    \end{equation}
    Following the arguments in \cite[Appendix~C]{li2007range}, it can then be shown that the columns of the precoder $\FF$ corresponding to the optimal $\XXbstar = \FF \FFh$ of \eqref{eq_problem_peb2} belong to the subspace spanned by the columns of $\UU$, i.e., $\projnull{\UU} \FF = 0$, 
    where $\projnull{\UU} = \Imatrix - \projrange{\UU}$ and $\projrange{\UU} =  \UU ( \UU^H \UU )^{-1} \UU^H$. 
    Then, the optimal precoder covariance matrix can be expressed as
    \begin{equation}
        \XXbstar = \projrange{\UU} \FF \FFh \projrange{\UU} = \UU \Lambdab \UU^H
    \end{equation}
    where $\Lambdab = ( \UU^H \UU )^{-1} \UU^H \FF \FFh \UU ( \UU^H \UU )^{-1}$.
    \qed
    \vspace{-0.1in}
\end{appendices}


\bibliographystyle{IEEEtran}
\bibliography{IEEEabrv,bib/5g_pos}

\end{document}